\documentclass[a4paper,12pt]{scrartcl}
\usepackage[centertags]{amsmath}
\usepackage{amsfonts}
\usepackage{amssymb}
\usepackage{amsthm}
\usepackage{nicefrac}
\usepackage{graphicx}

\theoremstyle{plain}
\newtheorem{thm}{Theorem}[section]
\newtheorem{prop}{Proposition}[section]

\begin{document}

\title{On the age-, time- and migration dependent dynamics of diseases}
\author{Ralph Brinks\\
Institute for Biometry and Epidemiology\\German Diabetes Center\\
Duesseldorf, Germany}
\date{}
\maketitle

\begin{abstract}
This paper generalizes a previously published differential
equation that describes the relation between the age-specific
incidence, remission, and mortality of a disease with its
prevalence. The underlying model is a simple compartment model
with three states (illness-death model). In contrast to the former
work, migration- and calendar time-effects are included. As an
application of the theoretical findings, a hypothetical example of
an irreversible disease is treated.
\end{abstract}


\section{Introduction}\label{intro}
With a view to basic epidemiological parameters such as incidence,
prevalence and mortality of a disease, it has been proven useful to
consider simple illness-death models \cite{Kei90} as shown in
Figure \ref{fig:CompModel}. Depending on the context, sometimes
these are referred to as state models or compartment models. Here
we consider three states: Normal or non-diseased with number of
people denoted as $S$ (susceptible), the diseased state with
number $C$ (cases) and the death state.

\begin{figure}[ht]
\centering
\includegraphics[keepaspectratio,width=0.85\textwidth]{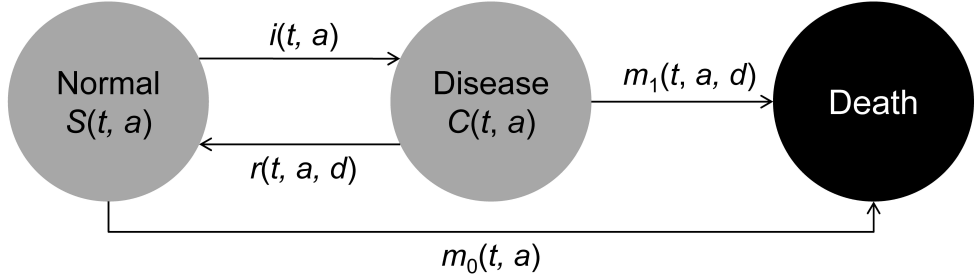}
\caption{Simple illness-death model} \label{fig:CompModel}
\end{figure}

\bigskip

The transition intensities between the states henceforth are
denoted with the symbols as in Figure \ref{fig:CompModel}:
incidence $i$, remission $r$ and mortality rates\footnote{The
expressions \emph{rate} and \emph{density} are synonymously used
in this article.} $m_0$ and $m_1$. In general, the intensities
depend on calendar time $t$, age $a$ and sometimes also on the
duration $d$ of the disease.

Models of this kind are quite common, see for example
\cite{Kei90}, \cite{Kei91} or the text book \cite{Kal02}. Murray
and Lopez (\cite{Mur94} and \cite{Mur96}) have considered such a
compartment model with rates being independent from calendar time
$t$ and duration $d$. In the context of the \emph{Global Burden of
Disease} study of the \emph{World Health Organization} they used
following system of ordinary differential equations (ODEs) to
describe the transitions between the three states:

\begin{equation}\label{eq:Murray}
\begin{split}
    \frac{\mathrm{d} S}{\mathrm{d} a} &= - \left ( i + m_0 \right ) \cdot S  + r \cdot C\\
    \frac{\mathrm{d} C}{\mathrm{d} a} &= i \cdot S - (m_1 + r) \cdot C.\\
\end{split}
\end{equation}

By this system the changes in the numbers of the non-diseased and
diseased persons aged $a$ are related to the intensities as in
Figure \ref{fig:CompModel}. Age plays here the role of temporal
progression. This homogeneous linear system of ODEs looks
relatively harmless, but is limited due to its heavy assumptions.
By an easy calculation it can be shown that Eq. \eqref{eq:Murray}
implies the population being stationary. Let $N(a) := S(a) + C(a)$
denote the total number of persons alive in the population aged
$a$. For $a \in [0, \omega]$ with $N(a) > 0$ define the
age-specific prevalence
\begin{equation}\label{eq:preval}
p(a) := \frac{C(a)}{C(a) + S(a)}.
\end{equation}

\noindent Then from Eq. \eqref{eq:Murray} it follows
\begin{align*}
\frac{\mathrm{d} N}{\mathrm{d} a} &= \frac{\mathrm{d} S}{\mathrm{d} a} + \frac{\mathrm{d} C}{\mathrm{d} a}\\
                                  &= -m_0 \cdot S - m_1 \cdot C\\
                                  &= - N \cdot \left [ (1-p) \cdot m_0 + p\cdot m_1 \right ].
\end{align*}
The term $(1-p) \cdot m_0 + p\cdot m_1$ is the overall mortality
$m$ in the population. Hence, it holds $\tfrac{\mathrm{d}
N}{\mathrm{d} a} = - m \cdot N$, which is the defining equation of
a stationary population, \cite{Pre82}. Although the model of
a stationary population is widely used in demography, real
populations merely are stationary. Moreover, the inclusion of the
values $S$ and $C$ is disturbing. It would be better if Eq.
\eqref{eq:Murray} could be expressed in terms of the age-specific
prevalence \eqref{eq:preval}, what indeed can be achieved. In
\cite{Bri11} it has been shown, that system \eqref{eq:Murray} can
be transformed into the following one-dimensional ODE of Riccati type:

\begin{equation}\label{eq:Brinks}
\frac{\mathrm{d} p}{\mathrm{d} a} = (1-p) \cdot \bigl ( i - p
\cdot \left (m_1 - m_0 \right ) \bigr ) - r \cdot p.
\end{equation}

The importance of Eq. \eqref{eq:Murray} and \eqref{eq:Brinks} is
obvious. For given incidence-, remission- and mortality-rates plus
an initial condition, the age profile of the numbers of patients
and the prevalence is uniquely determined, respectively. To state
it clearly, the ``forces'' incidence, remission and mortality
uniquely prescribe the prevalence - not only qualitatively but in
these quantitative terms. This is called the \emph{forward
problem}: we infer from the causes -- the forces -- to the effect
-- the numbers of diseased or the prevalence, respectively. If in
the scalar Riccati ODE \eqref{eq:Brinks} the age-profiles of the
prevalence, mortality and remission are known, one can directly
solve Eq. \eqref{eq:Brinks} for the incidence. This is the
\emph{inverse problem} -- we conclude from the effect to the
cause. This allows, for example, cross-sectional studies being
used for incidence estimates, where otherwise lengthy follow-up
studies are needed. For an example on real data, see \cite{Bri11}.
Recently, it has been proven that the inverse problem is ill-posed
\cite{Bri12}.

\bigskip

The article is organized as follows: In the next section Eq.
\eqref{eq:Brinks} is generalized allowing dependency on calendar
time and migration. The central result is a partial differential
equation (PDE). Similar to the ODE, in the general case there is a
forward and an inverse problem for the PDE, too. These are
analyzed in a simulated register data of a hypothetical chronic
disease in the section thereafter. Finally, the results are summed
up.

\section{General equation of disease dynamics}\label{sec:1}
In this section the simple illness-death model of Figure
\ref{fig:CompModel} is generalized. The rates $i, r, m_0$ and
$m_1$ henceforth depend on age $a$ and calendar time $t$, but are
assumed to be independent from the duration $d$. Furthermore, let
the numbers of the non-diseased $S(t, a)$ and diseased persons
$C(t, a)$ aged $a$ at time $t$ be non-negative and partially
differentiable. Define $N(t, a) := S(t, a) + C(t, a)$.
Additionally, let $\sigma(t, a)$ and $\gamma(t, a)$ denote those
proportions of $N(t, a)$, such that $\sigma(t, a) \cdot N(t, a)$
and $\gamma(t, a) \cdot N(t, a)$ are the net migration rates of
non-diseased and diseased persons aged $a$ at time $t$,
respectively:

\begin{equation}\label{eq:Big}
\begin{split}
    \left ( \frac{\partial}{\partial t} + \frac{\partial}{\partial a} \right ) S &= \sigma \cdot N - \left ( i + m_0 \right ) \cdot S  + r \cdot C\\
    \left ( \frac{\partial}{\partial t} + \frac{\partial}{\partial a} \right ) C &= \gamma \cdot N + i \cdot S - (m_1 + r) \cdot C.\\
\end{split}
\end{equation}

\noindent After introducing the age-specific prevalence $p(t, a)$
in year $t$, $$p(t, a) := \frac{C(t, a)}{C(t, a) + S(t, a)},$$ for
$(t, a) \in D := \{(t, a) \in [0, \infty)^2 ~| C(t, a) \ge 0,
~S(t, a) \ge 0, C(t, a) + S(t, a) > 0 \}$ the system
\eqref{eq:Big} can be transformed into an equation similar to
\eqref{eq:Brinks}:

\begin{thm}
Let $S(t, a)$ and $C(t, a)$ be given by Eq. \eqref{eq:Big}, then
$p(t, a)$ is partially differentiable in $D$ and it holds
\begin{equation}\label{eq:Brinks2}
   \left ( \frac{\partial}{\partial a} + \frac{\partial}{\partial t} \right ) p =
   (1-p) \left [ i - p (m_1 - m_0) \right ] - r p + \mu,
\end{equation}
where $\mu := \gamma (1-p) - p \sigma$ describes the impact of
migration.
\end{thm}
\begin{proof}
Follows from applying the quotient rule to $p(t, a) = \tfrac{C(t,
a)}{C(t, a) + S(t, a)}$ and using \eqref{eq:Big}.
\end{proof}

Obviously, if the incidence- and mortality rates do not depend on
the calendar time $t$, then from Eq. \eqref{eq:Brinks2} with $\mu
\equiv 0$ it follows \eqref{eq:Brinks}. Hence, Eq.
\eqref{eq:Brinks} does not depend on the stationary population
assumption.

\bigskip

For applications in epidemiology it is important that solutions of
Eq. \eqref{eq:Brinks2} are \emph{meaningful}, i.e. $p(t, a) \in
[0, 1]$ for all $(t, a) \in D$. Therefor we note:
\begin{thm}\label{t:equi}
For all $(t, a) \in D$ following statements are equivalent:
\begin{enumerate}
    \item [(1)] $p(t, a) = \tfrac{C(t, a)}{S(t, a) + C(t, a)}$ is a solution of
    Eq. \eqref{eq:Brinks2}.
    \item [(2)] $S(t, a) = \left ( 1-p(t, a) \right ) \cdot N(t, a)$ and
    $C(t, a) = p(t, a) \cdot N(t, a)$ are solutions to Eq.
    \eqref{eq:Big}.
\end{enumerate}
\end{thm}
\begin{proof}
This follows by inserting the expressions into the PDEs.
\end{proof}
By Theorem \ref{t:equi} a solution $p(t,a)$ of Eq.
\eqref{eq:Brinks2} can be written as $p(t, a) =
\tfrac{C(t,a)}{N(t,a)}$ with $N(t,a) = S(t,a) + C(t,a)$. For
$(t,a) \in D$ this implies $p(t,a) \in [0, 1].$

\bigskip

The migration term $\mu$ will be analyzed further now. Let
$\varphi := \sigma + \gamma$ be the overall migration rate. We
split all migration rates $f, ~f\in \{\varphi, \sigma, \gamma\}$
into a positive part $f_+ \ge 0$ (immigration) and a negative part
$f_- \ge 0$ (emigration):
$$f = f_+ - f_-, ~\text{for} ~f\in \{\varphi, \sigma, \gamma\}.$$

\noindent Moreover, for $\varphi_-(t, a) > 0$ define $p^{(m)}_-(t,
a) := \tfrac{\gamma_-(t, a)}{\varphi_-(t, a)}$ the prevalence of
the disease in the emigrants and for $\varphi_+(t, a)
> 0$ define $p^{(m)}_+(t, a) := \tfrac{\gamma_+(t, a)}{\varphi_+(t, a)}$
the prevalence in the immigrants.

\begin{prop}\label{prop:mu}
With the notations as above it holds
\begin{equation*}
\mu(t, a) = \left\{%
\begin{array}{ll}
    ~~\varphi_+(t, a) \cdot p^{(m)}_+(t, a) & \\
    - \varphi_-(t, a) \cdot p^{(m)}_-(t, a) - \varphi(t, a) \cdot p(t, a), & \hbox{for $\varphi_-(t, a), \varphi_+(t, a) > 0$;} \\
    ~~\varphi_+(t, a) \cdot \left [ p^{(m)}_+(t, a) - p(t, a) \right ], & \hbox{for $\varphi_-(t, a) = 0, \varphi_+(t, a) > 0$;} \\
    -\varphi_-(t, a) \cdot \left [ p^{(m)}_-(t, a) - p(t, a) \right ], & \hbox{for $\varphi_-(t, a) > 0, \varphi_+(t, a) = 0$;} \\
    ~~~~~~~~0, & \hbox{for $\varphi_-(t, a) = \varphi_+(t, a) = 0$.} \\
\end{array}%
\right.
\end{equation*}
\end{prop}
\begin{proof}
For all $(t, a) \in D$ it holds $\mu = \gamma - \varphi \cdot p.$
By splitting this expression into positive and negative parts, the
Proposition follows.
\end{proof}

With the assumption that the prevalence of those aged $a$ at time
$t$ who immigrate is the same of those who emigrate, say
$p^{(m)}(t, a)$, then it holds
\begin{equation}
    \mu = \varphi \left ( p^{(m)} - p \right ).
\end{equation}

Hence, if the prevalence $p^{(m)}$ of the migrants is the same as
of those who stay, $p^{(m)} \equiv p$, the change in prevalence $(
\tfrac{\partial} {\partial a} + \tfrac{\partial}{\partial t} ) p$
does not depend on migration.

This is an important result, because in illness-death models the
assumption of absence of migration is often made. In our framework
this restriction is not necessary. Even if there is migration, but
the prevalence in the migrants is the same as in the resident
population, then the prevalence is not affected by migration.

\bigskip

The solution of Eq. \eqref{eq:Brinks2} can be obtained by the
methods of characteristics \cite{Pol02}. Let an initial condition
of the form $p(a, 0) = p_0(a)$ be given, then we have a so called
Cauchy problem, which has a unique solution if the right-hand side
of the PDE is sufficiently smooth \cite{Pol02}. This solution is
calculated as follows. Assume, the prevalence for those aged
$\tilde{a}$ in year $\tilde{t}$ has to be calculated. First,
rearrange $( \tfrac{\partial} {\partial t} +
\tfrac{\partial}{\partial a} ) p$ such that
\begin{equation*}
\left ( \frac{\partial}{\partial a} + \frac{\partial}{\partial t}
\right ) p = \alpha(t, a) + \beta(t, a) ~p + \gamma(t, a) ~p^2.
\end{equation*}

Second, solve the initial value problem given by following Riccati
ODE:
\begin{equation}\label{eq:PDE2ODE}
    \frac{\mathrm{d}y(\tau)}{\mathrm{d}\tau}
    = \alpha(\tau + a_0, \tau)
      + \beta(\tau + a_0, \tau) ~y
      + \gamma(\tau + a_0, \tau) ~y^2,
\end{equation}
and initial value $y(0) = p_0(a_0)$ where $a_0 := \tilde{a} -
\tilde{t}.$ Then, an easy calculation shows that $y(\tilde{t}) =
p(\tilde{t}, \tilde{a})$ is the desired value.


\section{Application on a simulated register}
In this section, the application of the above-formulated Cauchy
problem on a simulated register of a chronic disease is shown.
Since the disease is assumed to be irreversible, it holds $r
\equiv 0.$ First, we address a direct problem: From given
age-specific prevalence in some point in time $t_0$ we want to
deduce the age-specific prevalence in $t_1, ~t_1
> t_0,$ by applying Eq. \eqref{eq:Brinks2} with $\mu \equiv 0$.
Second, an inverse problem is formulated. Assume in the year $t_0$
the functions $p_0 = p (t_0, \cdot)$, $i(t_0, \cdot)$, $m_0 (t_0,
\cdot)$ and $m_1(t_0, \cdot)$ were measured. If in year $t_1, ~
t_1 > t_0,$ the age profile of the prevalence $p(t_1, \cdot)$ is
given, the question arises: how has the course of the age-specific
incidence changed in the meantime? This is an inverse problem,
because we infer from the effect (prevalence in $t_1$) on the
causes. Here we will formulate a simple, straight forward solution
by an optimization approach.

Both problems will be treated based on data of a simulated
register. The register is designed such that in a period of 150
years all persons are tracked from birth to death. For each
person, the date of an eventual diagnosis of the chronic disease
is recorded. For the simulation, the following assumptions are
placed as a basis:

\begin{enumerate}
   \item In each calendar year 0 to 150 2,000 people are born. The
      births during the calendar year follow a uniform distribution.
   \item The mortality of the non-diseased persons is of Strehler-Mildvan type
      and is given by the equation $$m_0 (t, a) = \exp (-10.7 + 0.1 a) \cdot
      (1-0.002)^{(t-20)_+}.$$ The notation $(t-20)_+$ denotes the
      positive component of the expression $(t-20)$. The exponential
      term approximates the current mortality of men in Germany, the
      second factor takes the increasing life expectancy into account.
   \item The incidence is described by the equation
   \begin{equation}\label{eq:inc}
   i(t, a) = \frac{(a - 30)_+}{3000} \cdot 0.99^{(t-50)_+}.
   \end{equation}
   \item The relative risk of death is constant for all ages and times:
      $$R (t, a) = \frac{m_1 (t, a)}{m_0 (t, a)} = 2.$$
\end{enumerate}

\bigskip

\noindent After the simulation, each person in the register is
represented by four pieces of information:
\begin{enumerate}
  \item[1)] A unique identification number (an integer),
  \item[2)] Calendar year of birth,
  \item[3)] The person's age in years at diagnosis (0 if the person does not
  fall ill),
  \item[4)] Age of death of the person in years.
\end{enumerate}

Entries 2) - 4) in the register are given to three decimals, which
corresponds to a precision of one day. The identification number
of the person is an ongoing counter. The date of birth (in
calendar years) is given by the simulated year, the decimals are
drawn from a uniform distribution $[0, 1).$ To decide if a thus
far non-diseased person born in year $\tau$ becomes ill or dies
without the disease, a competing risk approach in a discrete event
simulation (DES) is accomplished. Based on the cumulative
distribution function of the common risk (total intensity
$i(\tau+a, a) + m_0 (\tau+a, a)$), the age $a_0$ of event is drawn
by the inverse transform sampling (inversion method,
\cite{Gil11}). Based on a comparison between $i (\tau + a_0, a_0)$
and $m_0 (\tau + a_0, a_0),$ it is decided whether the onset of
the disease or the death without disease occurred. In the first
case $a_0$ represents the age at disease's onset, in the second
case, $a_0$ is the age of death. If the person gets the disease,
the age of death is simulated (conditional on reaching the age
$a_0$).

As in the calendar years 0 to 150 exactly 2000 people are born
every year, the hypothetical register contains $151 \cdot 2000 =
302,000$ persons. Then, the events of the hypothetical register
are transformed into a Lexis diagram of five years intervals
\cite{Lex03}. This allows an easy extraction of the
person-years and the numbers of events in the corresponding age-
and period classes.

\bigskip

In both test cases, in the direct and the inverse problem, we
assume information to be given only in two points in time, $t_0$
and $t_1$. Of course, three or more points in time would be
advantageous, but with respect to applicability in epidemiological
contexts, the test problems try to mimic a minimalistic
setting.

\subsection{Direct problem}\label{ss:direct}
Assume we have measured the age profile $p_0 = p (t_0, \cdot)$ of
the prevalence in $t_0$, and the age-specific incidence $i(t_0,
\cdot)$ and mortality densities $m_0 (t_0, \cdot)$ and $m_1(t_0,
\cdot)$. Furthermore, at a later point in time $t_1 > t_0$ let the
age-specific rates $i(t_1, \cdot)$ and mortalities $m_0(t_1,
\cdot)$ and $m_1 (t_1, \cdot)$ be given. The direct problem refers
to the question: what can be said about the age-specific
prevalence $p(t_1, \cdot)$ in $t_1$?

To answer this question, age-specific incidence and mortality
rates at two time points $t_0 = 120$ and $t_1 = 140$ (years) are
extracted from the register. In addition, the age-specific
prevalence $p(t_0, \cdot)$ is collected at $t_0$. Figure
\ref{fig:inc} shows the extracted age-specific incidence density
(dashed lines) in $t_0$ (red) and $t_1$ (blue) in comparison with
the theoretical values (solid lines).

\begin{figure}[ht]
  \centering
  \includegraphics[width=0.75\textwidth]{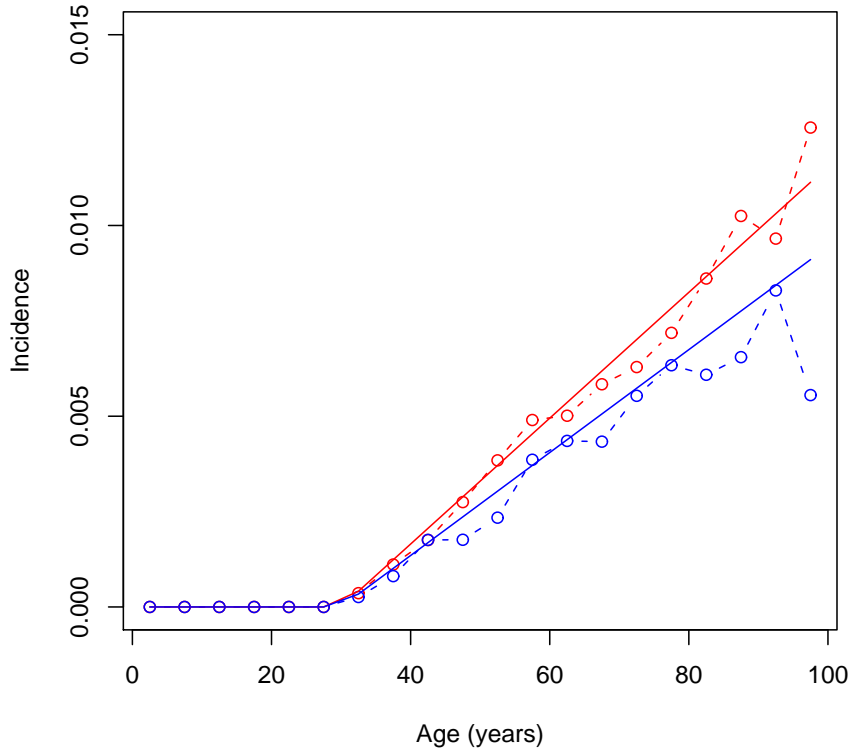}
\caption{Age-specific incidence density extracted from the
register (dashed lines) at $t_0 = 120$ (red) and $t_1 = 140$
(blue) in comparison with the theoretical values (solid lines).}
\label{fig:inc}
\end{figure}

Now consider the Cauchy problem that is given by Eq.
\eqref{eq:Brinks2} with the initial condition $p (t_0, \cdot) =
p_0$. For the solution one needs the functions $i(t, \cdot)$, $m_0
(t, \cdot)$ and $m_1(t, \cdot)$ for all time points $t$ between
$t_0$ and $t_1$. For this, the function values are interpolated
affine-linearly. The initial value problem Eq. \eqref{eq:PDE2ODE}
is solved numerically using the MATLAB\footnote{The MathWorks,
Natick, Massachusetts, USA} function \verb"ode45".

If we compare the numerical solution of the Cauchy problem in year
$t_1 = 140$ with the actually observed prevalence in the year 140,
one gets the result as shown in Figure \ref{fig:DirectSolution}.

\begin{figure}[ht]
  \centering
  \includegraphics[width=0.75\textwidth]{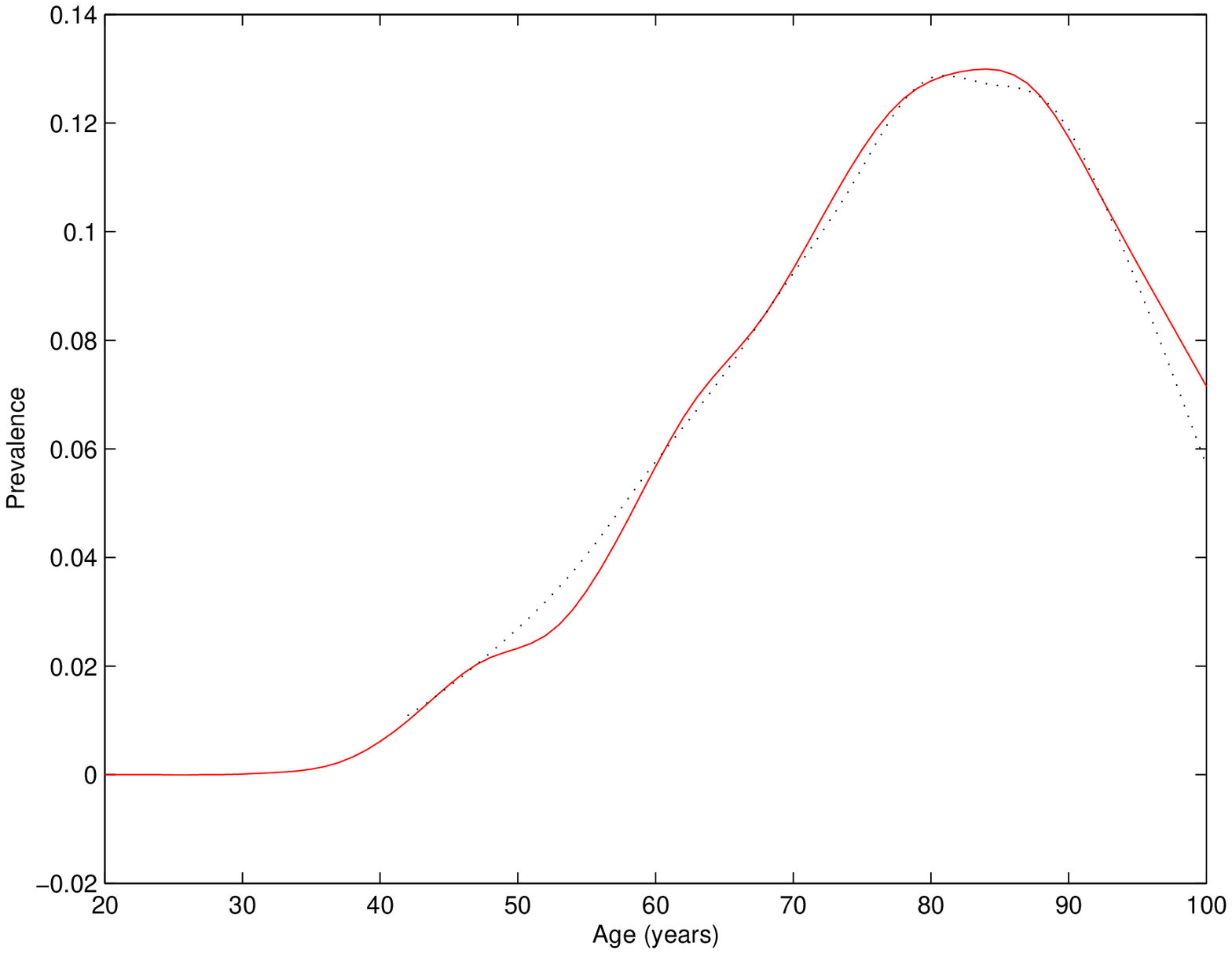}
\caption{Numerical solution of the direct problem (dashed line)
compared to the observed prevalence in year $t_1 = 140$ (solid
line).} \label{fig:DirectSolution}
\end{figure}

Visually this gives a fairly good agreement between the predicted
curve with the actually observed age-specific prevalence. The
maximum absolute deviation is 0.0146, which means that in this
example the prevalence can be predicted up to 1.5 percent points.
The largest deviation is in the oldest age class, when we have
only a few cases of the disease.

\subsection{Inverse problem}
In epidemiological studies, it is more laborious to measure
incidence rates than prevalences. Hence, in practice, the
following inverse problem is much more important than the direct
problem of the previous section. Assume in the year $t_0 = 120$
the functions $p_0 = p (t_0, \cdot)$, $i(t_0, \cdot)$, $m_0(t_0,
\cdot)$ and $m_1 (t_0, \cdot)$ are known. Moreover, in year $t_1 =
140$ let the age profile of the prevalence $p(t_1, \cdot)$ be
given. The functions $m_0 (t_1, \cdot)$ and $m_1(t_1, \cdot)$ are
also assumed to be known (for example from other epidemiological
studies). The question then is, how well the incidence $i(t_1,
\cdot)$ in the year $t_1$ can be derived from this information.
For simplicity, we assume that the incidence of $i(t_1, \cdot)$ in
$t_1$ can be expressed as a product

\begin{equation}\label{eq:i1}
i (t_1, \cdot) = i (t_0, \cdot) \cdot (1 - h),
\end{equation}

where $h \in [0, 1].$ The upper limit for $h$ stems from the fact,
that incidence rates are non-negative. The lower limit reflects
the prior knowledge, that incidence has not increased in $t_1$
compared to $t_0$: $i(t_1, a) \le i (t_0, a),$ for all $a \in [0,
\infty).$ Equation \eqref{eq:i1} corresponds to a proportional
hazards approach, which is used widely in epidemiology.

To solve this inverse problem, we formulate an optimization
problem. For given $h \in [0, 1]$ and $i(t_0, \cdot)$ by Eq.
\eqref{eq:i1} the function $i(t_1, \cdot)$ is defined. If
furthermore $p (t_0, \cdot)$, $m_0(t_0, \cdot), ~m_0(t_1, \cdot),
m_1(t_0, \cdot)$ and $m_1 (t_1, \cdot)$ are known, then we are in
the situation to calculate a unique function $\hat{p}_h(t_1,
\cdot)$ by solving the Cauchy problem described in the previous
subsection \ref{ss:direct}. The solution $\hat{p}_h(t_1, \cdot)$
of the direct problem depends on $h.$ We can compare
$\hat{p}_h(t_1, \cdot)$ with the measured prevalence $p(t_1,
\cdot)$ in the register. Thus, we seek for $h^\ast \in [0, 1]$
that minimizes the Euclidean distance between $\hat{p}_h(t_1,
\cdot)$ and $p(t_1, \cdot):$

\begin{equation}\label{eq:Euclid}
h^\ast = \arg \min_{h \in (0, 1)} \int_{A_{t_1}} | \hat{p}_h(t_1,
a) - p(t_1, a)|^2 \mathrm{d}a,
\end{equation}
where $A_{t_1} = \{ a \in [0, \infty) ~| (t_1, a) \in D\}$.

Figure \ref{fig:3} shows the Euclidean distance between the
prevalence $p(t_1, \cdot)$ in the register and the solution
$\hat{p}_h(t_1, \cdot)$ as a function of $h$.

\begin{figure}[ht]
  \centering
  \includegraphics[width=0.75\textwidth]{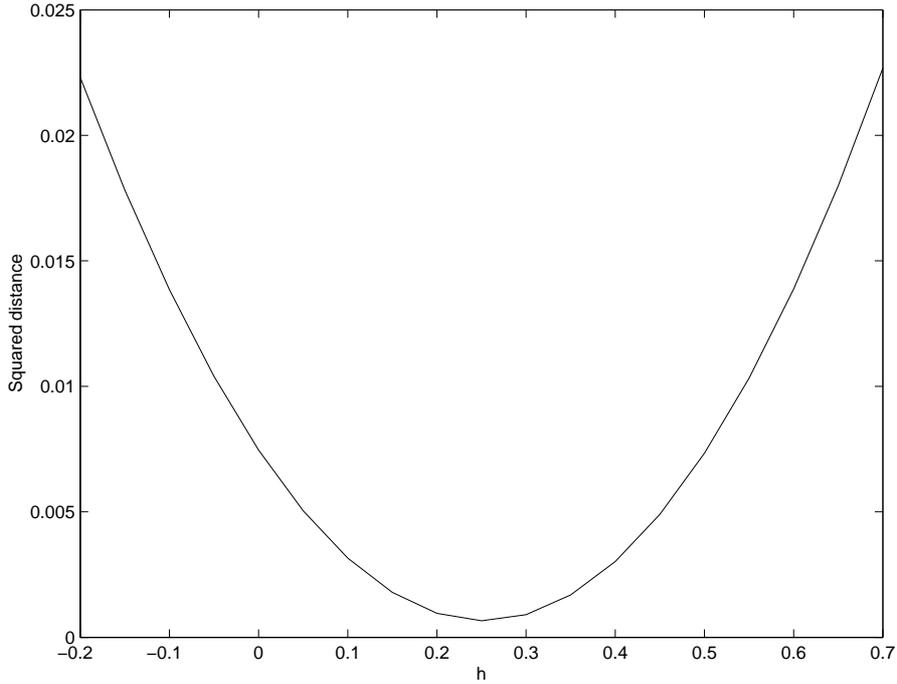}
\caption{Euclidean distance (as in Eq. \eqref{eq:Euclid}) as a function
of $h$. There is a unique minimum $h^\ast = 0.25$.} \label{fig:3}
\end{figure}

From the graph in Figure \ref{fig:3} it is obvious that the square
of the distance is minimized at about $h^\ast = 0.25$. Since from
the 50th calendar year the incidence decreases by 1\% per year and
a period of 20 years was considered, a factor $1-h$ of about
$0.99^{20} = 0.82 = (1 - 0.18)$ is expected. The revealed value
$h^\ast = 0.25$ is about a factor of 1.4 too large.

\section{Discussion}
In this work we developed a new equation linking incidence-,
remission- and mortality-rates with prevalence of a disease. In
contrast to former works, the assumptions of stationary
populations, independence from calendar time and zero net
migration have been released. The new equation has a wide range
of applicability in epidemiological, health care and health
economic contexts.

However, it has several limitations. First, Eq. \eqref{eq:Brinks2}
needs the remission rate $r$ and mortality rate $m_1$ of the
diseased to be independent from the duration $d$ of the disease.
In real diseases independence from duration is only an
approximation. For many infectious diseases, immune response is
dependent on the time since onset of the disease. Also in chronic
diseases duration since onset plays a major role. For example, the
age- and sex-adjusted mortality due to coronary heart disease
roughly doubles for each 10-year increase in diabetes duration.
The all-cause mortality increases by a factor of 1.2 per 10-year
duration, \cite{Fox04}.

Second, although the new equation is not limited to the case $\mu
\equiv 0$, in practical applications information about the health
of immigrants \emph{and} emigrants is seldom obtainable. By
Proposition \ref{prop:mu} reasonable knowledge of prevalence in
all migrants is necessary to accurately treat the case $\mu \neq
0$. To give an example, countries with large-scale immigration
programs such as Canada observe a so-called \emph{healthy
immigrant effect} with respect to chronic diseases: immigrants are
healthier than residents, \cite{McD04}. Assumed that the emigrants
from Canada have the same prevalence as the residents, it would
follow $\mu \neq 0.$ However, surveys about the health status of
emigrants are missing. The reason is obvious, Canada's
taxpayer-funded health care system is interested in measuring
health of those who immigrate, but not in those who emigrate.
Hence, information is lacking and assumptions have to be made.

Third, Eq. \eqref{eq:Brinks2} only considers prevalence in
migrants at the moment of emigration or immigration. Of course,
large scale immigration is likely to change the incidence of the
disease in the population, because immigrants' health adapts to
the new environment. There are many examples where
immigrants from the developing countries increase incidence of
diabetes and related complications when adopting westernized
lifestyle, \cite{Mis07}. The opposite may also be true, in Canada
immigrants continue to have a lower relative risk of chronic
conditions compared to the native-born, even many years after
immigration, \cite{McD04}.

\bigskip

Beside theoretical considerations, we use the new equation in a
simulated register of a hypothetical chronic disease. The register
has been simulated by Monte Carlo techniques and has been analyzed
by a numerical implementation of the new equation. To check the
practical applicability of the analysis, the simulation and the
analysis have been strictly separated, i.e. neither was the PDE
used in simulating the register, nor was information other than
explicitly mentioned, used as input for the simulation exploited
in the analysis. The PDE has only been used in the analysis of the
direct and inverse problem. In the direct problem, the prevalence
at the later point in time $t_1$ could be predicted from the
prevalence in $t_0$ twenty years earlier with a high accuracy. Of
course, the obtained accuracy is a result of the structure
inherent to the simulation. The solution of both, direct and
inverse problem, uses an affine-linear interpolation for the
incidence- and mortality rates between $t_0$ and $t_1$. In the
simulated register this works well, because it reflects the trends
in the incidence and mortalities. Affine-linear interpolation will
impose problems if the incidence trend turns around between $t_0$
and $t_1$. An example for a change of trends can be found in
\cite{Car08}: from 1995 to 2004 incidence of diabetes is found to
be rising with an average of 5.3\% per year in all age classes,
and from 2005 to 2007 incidence is declining with 3.1\% per year.

In the inverse problem, the incidence in $t_1$ was reconstructed
from the observed prevalence in $t_1$. Provided that the
right-hand side of Eq. \eqref{eq:Brinks2} is sufficiently smooth,
existence of $h^\ast \in [0, 1]$ follows from the continuous
dependency of the solution of the Cauchy problem on $h$ from the
compact interval $[0, 1]$. Continuous dependency can be seen by
noting that the solution constructed by the methods of
characteristics inherits its smoothness properties from the
smoothness of the right-hand side of Eq. \eqref{eq:PDE2ODE},
\cite{Hal80}. The question remains why the result (in terms of
$h^\ast$) is about a factor 1.4 too large. The approach in solving
the inverse problem is the proportional hazards assumption Eq.
\eqref{eq:i1}. Indeed, the simulation considers a decline of
exponential type, see Eq. \eqref{eq:inc}. Although the exponential
in this case can approximated by an affine-linear interpolation
function quite well, it appears that the solution of the inverse
problem reacts quite sensitively on inaccuracies. This is in line
with our observation, that the inverse problem is ill-posed
\cite{Bri12}.


\begin{thebibliography}{}
\bibitem{Bri11}
Brinks R (2011). A new method for deriving incidence rates from
prevalence data and its application to dementia in Germany.
\verb"http://arxiv.org/abs/1112.2720v1"

\bibitem{Bri12}
Brinks R (2012). On characteristics of an ordinary differential
equation and a related inverse problem in epidemiology.
\verb"http://arxiv.org/abs/1208.5476v1"

\bibitem{Car08}
Carstensen B, Kristensen JK, Ottosen P, Borch-Johnsen K (2008).
The Danish NationalDiabetes Register: trends in incidence,
prevalence and mortality. Diabetologia 51 (12):2187-96.

\bibitem{Fox04}
Fox CS, Sullivan L, D'Agostino RB Sr, Wilson PW (2004) The
significant effect of diabetes duration on coronary heart disease
mortality: the Framingham Heart Study. Diabetes Care 27 (3):704-8.

\bibitem{Gil11}
Gilli M, Maringer D, Schumann E (2011) Numerical Methods and
Optimization in Finance. Academic Press, Waltham, MA.


\bibitem{Hal80}
Hale JK (1980) Ordinary differential equations. Robert Krieger
Publishing, Malabar, FL.

\bibitem{Kal02}
Kalbfleisch JD, Prentice RL (2002) The Statistical Analysis of
Failure Time Data, 2nd edn. John Wiley \& Sons, Hoboken, NJ.

\bibitem{Kei90}
Keiding N, Hansen BE, Holst C (1990) Nonparametric estimation of
disease incidence from a cross-sectional sample of a stationary
population. Lect Notes Biomath 86: 36-45.

\bibitem{Kei91}
Keiding N (1991) Age-specific incidence and prevalence: a
statistical perspective. J R Statist Soc A 154:371-412.

\bibitem{Lex03}
Lexis W (1903) Abhandlungen zur Theorie der Bev\"olkerungs- und
Moralstatistik. Gustav Fischer, Jena,
\verb"http://dspace.utlib.ee/dspace/bitstream/10062/5316/4/lexis_abhandlocr.pdf" Accessed 14 August 2011.

\bibitem{McD04}
McDonald JT, Kennedy S (2003) Insights into the healthy immigrant
effect: health status and health service use of immigrants to
Canada. Soc Sci \& Med 59: 1613-27.

\bibitem{Mis07} Misra A, Ganda OP (2007) Migration and its impact on adiposity
and type 2 diabetes. Nutrition 23 (9):696-708

\bibitem{Mur94}
Murray CJL, Lopez AD (1994) Quantifying disability: data, methods
and results. Bull World Health Organ 72 (3):481-494.

\bibitem{Mur96}
Murray CJL, Lopez AD (1996) Global and regional descriptive
epidemiology of disability: incidence, prevalence, health
expectancies and years lived with disability. In: Murray CJL,
Lopez AD (ed) The Global Burden of Disease. Harvard School of
Public Health, Boston, pp 201-46.

\bibitem{Pol02}
Polyanin AD, Zaitsev VF, Moussiaux A (2002) Handbook of First Order
Partial Differential Equations. Taylor \& Francis, London.

\bibitem{Pre82}
Preston SH, Coale AJ (1982) Age structure, growth, attrition, and
accession: A new synthesis. Pop Index 48 (2):217-59.
\end{thebibliography}
\end{document}